\documentclass[10pt,journal]{IEEEtran}
\IEEEoverridecommandlockouts
\usepackage{epsfig,amsmath,amsthm,algorithm,amssymb,epsf,cite,subfigure,scalefnt,graphicx,setspace,color}
\usepackage[noend]{algpseudocode}
\newcommand\disequal{\stackrel{\text{d}}{=}}

   \def\cD{{\mathcal{D}}}

  \def\cS{{\mathcal{S}}} 
\def\cU{{\mathcal{U}}} \def\cV{{\mathcal{V}}}

\def\ba{{\mathbf{a}}} \def\bb{{\mathbf{b}}}

\def\bp{{\mathbf{p}}} \def\bq{{\mathbf{q}}} \def\br{{\mathbf{r}}}  
\def\bu{{\mathbf{u}}}   \def\bx{{\mathbf{x}}} \def\by{{\mathbf{y}}}
\def\bz{{\mathbf{z}}}

\def\bA{{\mathbf{A}}}

\def\argmin{\mathop{\mathrm{argmin}}}
\def\argmax{\mathop{\mathrm{argmax}}}

\def\mod{\mathop{\mathrm{mod}}}

     \def\d4{\!\!\!\!}              
                       \def\al{\alpha}
\def\Del{\Delta}

  \def\R{{\mathbb{R}}} \def\C{{\mathbb{C}}}   \def\E{{\mathbb{E}}}

\def\lp{\left(}     \def\rp{\right)}

  \def\-{\! - \!}  \def\+{\! + \!}  \def\={\! = \!}  \def\>{\! > \!} \def\nn{\nonumber}

\newtheorem{theorem}{Theorem}
\newtheorem{lemma}{Lemma}

\newtheorem{remark}{Remark}

\newcommand{\bef}{\begin{figure}}
\newcommand{\eef}{\end{figure}}
\newcommand{\beq}{\begin{eqnarray}}
\newcommand{\eeq}{\end{eqnarray}}

\title{Coherence Statistics of Structured Random Ensembles and Support Detection Bounds for OMP}
\author{Qiyou~Duan,
        Taejoon~Kim,~\IEEEmembership{Member,~IEEE,}
        Lin~Dai,~\IEEEmembership{Senior Member,~IEEE,}
        and~Erik~Perrins,~\IEEEmembership{Senior Member,~IEEE}
	\thanks{ Q. Duan and L. Dai are with the Department of Electrical Engineering, City University of Hong Kong, Hong Kong (e-mail: qyduan.ee@my.cityu.edu.hk; lindai@cityu.edu.hk).
T. Kim and E. Perrins are with the Department of Electrical Engineering and Computer Science, The University of Kansas, Lawrence, KS 66045 USA (e-mail: taejoonkim@ku.edu; esp@ieee.org). }
}

\begin{document}

\maketitle

\begin{abstract}
A structured random matrix ensemble that maintains constant modulus entries and unit-norm columns, often called a random phase-rotated (RPR) matrix, is considered in this paper. We analyze the coherence statistics of RPR measurement matrices and apply them to acquire probabilistic performance guarantees of orthogonal matching pursuit (OMP) for support detection (SD). It is revealed via numerical simulations that the SD performance guarantee provides a tight characterization, especially when the signal is sparse.
\end{abstract}



\section{Introduction} \label{Sec.1} 
Random matrix ensembles have found wide applications in fields of wireless communications and signal processing \cite{Liang07,Menon12,Elkhalil18,Zhang18}.
Despite the fact that most studied Gaussian measurement ensembles offer trackable analyses and appealing results \cite{Tropp07,Fletcher12,Lee16}, they are of somewhat limited use in practical applications because the design of measurement matrices is usually subject to physical or other constraints provided by a specific system architecture.
It is desirable to explore random matrix ensembles with hidden structure from a computational and an application-oriented point of view.

Coherence has been utilized to measure the quality of the measurement matrix \cite{Donoho01}. Analysis of coherence statistics of random vectors/matrices plays an important role in solving a series of signal processing problems including the Grassmannian line packing \cite{Love03,Mukkavilli03}, random vector quantization \cite{Jindal06,Au-yeung07}, and support detection (SD) \cite{Tropp07,Ben-Haim10,Bracher12,Malhotra17}. In particular, the performance of SD considerably varies with the characteristics of measurement matrices. There is a certain class of random matrix ensembles with hidden structures that can demonstrate an improvement in SD performance guarantees compared to Gaussian ensembles \cite{Tropp07}. Distinguished from the Gaussian measurement matrix that does not contain hidden constraints, the random phase-rotated (RPR) measurement matrix, where each entry is drawn from the constant modulus uniform phase rotation distribution, brings the benefits of maintaining unit-norm columns and constant modulus entries of the measurement matrix. This measurement ensemble has been utilized in advanced beamforming and precoding for wireless communications \cite{Hur13,Kim15}.

In this paper, we calculate high probability bounds on the coherence statistics of RPR measurement matrices and apply them to obtain SD performance guarantees for orthogonal matching pursuit (OMP), which is a low-complexity, greedy approach for SD \cite{Cai11,Tropp07}. The performance bound is in terms of the required number of measurements for any given number of supports and system dimensions. A free variable is introduced, which is optimized to further tighten the performance bound.
The main motivation is that previous work relying on the coherence property did not contain hidden constraints that are suitable for SD of OMP.
Numerical evaluations demonstrate that the analyzed SD performance guarantee of OMP is tight, especially when the signal is sparse.


\section{Coherence Statistics} \label{Sec.2} 
Suppose a random measurement matrix $\bA=[\ba_1,\ba_2,\cdots,$ $\ba_N]\in\C^{M\times N}$ with $\ba_n \in \C^{M\times 1}$ being the $n$th column of $\bA$. Each entry of $\bA$ is constant modulus and drawn from the random phase rotation variable as
\beq \label{random phase-rotated measurements}
A_{mn}=\frac{1}{\sqrt{M}}e^{j\Theta_{mn}},
\eeq
where $A_{mn}$ denotes the $m$th row and $n$th column entry of $\bA$, $m=1,\ldots,M$, $n=1,\ldots,N$, and the phase $\Theta_{mn}$ is an independent and identically distributed (i.i.d.) uniform random variable, i.e., $\Theta_{mn}\sim\cU[0,2\pi)$. With the construction in \eqref{random phase-rotated measurements}, $\bA$ maintains $\|\ba_n\|=1$, $\forall n$.

The coherence of $\bA$ is the maximum absolute correlation between two distinct columns of $\bA$ \cite{Tropp04}, which is given by
\beq \label{mutual coherence}
\mu(\bA) \triangleq \max_{i\neq j} |\ba_i^*\ba_j|,
\eeq
where $(\cdot)^*$ denotes the conjugate transpose.
Characterizing the distribution of $\mu(\bA)$ is of interest - however, it is challenging to directly derive the distribution of $\mu(\bA)$ when $\bA$ follows \eqref{random phase-rotated measurements}.
To circumvent this difficulty, we relegate to find a lower bound on the cumulative distribution function (CDF) of $\mu(\bA)$ instead.
We start by building a connection between the vector drawn from the distribution in \eqref{random phase-rotated measurements} and the vector consisting of Bernoulli random variables.

\begin{lemma}\label{RPR and Bernoulli vectors}
Let $\bp\in\C^{M\times1}$ and $\bq\in\R^{M\times1}$ be random vectors with i.i.d. entries $p_m=1/\sqrt{M}e^{j\theta_m}$, $\theta_m\in \cU[0,2\pi)$, and $q_m\in\{-1/\sqrt{M},1/\sqrt{M}\}$ with equal probability for $m=1,\ldots,M$, respectively. Then, for any unit-norm vector $\bu\in\C^{M\times1}$, the following inequality holds
\beq \label{inequality between RPR and Bernoulli vectors}
\mathbb{E}\left[|\bp^*\bu|^{2k}\right] \leq \mathbb{E}\left[|\bq^*\bar{\bu}|^{2k}\right],
\eeq
where $\bar{\bu}\in\R^{M\times1}$ has each entry $\bar{u}_m=1/\sqrt{M}$, $\forall m$, $k$ is a nonnegative integer, and the expectations are taken over $\bp$ and $\bq$, respectively.
\end{lemma}

\quad\textit{Proof:} See Appendix \ref{Proof of Lemma 1}.

Based on Lemma \ref{RPR and Bernoulli vectors}, we characterize a bound on the distribution of $|\bp^*\bu|$ below.
\begin{lemma} \label{Probability Upperbound for inner product}
  Suppose the vectors $\bp$ and $\bu$ defined in Lemma \ref{RPR and Bernoulli vectors}. Then, for any $\delta>0$, the following inequality holds
  \beq \label{upper bound for inner product}
  \Pr(|\bp^*\bu|\geq \delta) \leq \Big(1-\frac{2}{g}\Big)^{-\frac{1}{2}}e^{-\frac{\delta^2M}{g}},\ g>2.
  \eeq
\end{lemma}

\quad\textit{Proof:} See Appendix \ref{Proof of Lemma 2}.

\begin{remark}
It is also possible to derive an upper bound on $\Pr(|\bp^*\bu|\geq \delta)$ by leveraging the matrix Bernstein inequality \cite[Theorem 1.6.2]{Tropp15}, which leads to $\Pr(|\bp^*\bu|\geq \delta)\leq 4e^{-\frac{3M\delta^2}{2\delta\sqrt{M}+6}}$. However, this bound is looser than that in \eqref{upper bound for inner product}.
\end{remark}

A lower bound on the CDF of $\mu(\bA)$ in \eqref{mutual coherence} can be found.
\begin{theorem} \label{lower bound of coherence statistic of RPR}
  Suppose a matrix $\bA\!\in\!\C^{M\times N}$ consisting of i.i.d. entries $A_{mn}\=1/\sqrt{M}e^{j\Theta_{mn}}$, $\Theta_{mn}\in\cU[0,2\pi)$, $m=1,\ldots,M$, $n=1,\ldots,N$. Then, the following holds for $g>2$,
  \beq \label{Coherence Statistic of RPR}
  \Pr(\mu(\bA)<\delta)\geq \bigg(1-\Big(1-\frac{2}{g}\Big)^{-\frac{1}{2}}e^{-\frac{\delta^2M}{g}}\bigg)^{\frac{N(N-1)}{2}}.
  \eeq
\end{theorem}
\begin{proof}
The inner product between two distinct column vectors of $\bA$ satisfies
  \beq
  \ba_{n_1}^*\ba_{n_2}=\sum_{m=1}^{M}\frac{1}{M}e^{j\Del\Theta_m}\disequal\sum_{m=1}^{M}\frac{1}{M}e^{j\xi_m} = \bp^*\bar{\bu}, \label{inner product for MC}
  \eeq
  where $\Del\Theta_m\triangleq \Theta_{mn_2}-\Theta_{mn_1}$, $n_1\neq n_2$, is the difference between two independent uniform random variables, whose probability density function is given by
  \beq
  p(\Del\Theta_m) = \left\{ \begin{array}{ll}
                         \frac{2\pi-|\Del\Theta_m|}{4\pi^2}, & \mbox{if } -2\pi\leq\Del\Theta_m<2\pi \\
                         0, & \mbox{otherwise}.
                       \end{array} \right.   \nonumber
  \eeq
  In \eqref{inner product for MC}, $\bar{\bu}$ follows the same definition in Lemma \ref{RPR and Bernoulli vectors}, and we use the fact that $e^{j\Del \Theta_m}=e^{j\mod(\Del \Theta_m,2\pi)}$ and $\xi_m\triangleq\mod(\Del \Theta_m,2\pi)$, in which $\mod(a,b)$ is the modulo $b$ of $a$.
  Note that $\xi_m\sim\cU[0,2\pi)$ and it verifies that $\ba_{n_1}^*\ba_{n_2}$ has the same distribution as $\bp^*\bar{\bu}$ in \eqref{inner product for MC}, where $\disequal$ is the equality in distribution.

  By Lemma \ref{Probability Upperbound for inner product}, we now have
  $\Pr(|\ba_{n_1}^*\ba_{n_2}|<\delta)=\Pr(|\bp^*\bar{\bu}|<\delta)\geq 1-(1-2/g)^{-1/2}e^{-\delta^2M/g}$.
  Then, the maximum order statistic of $|\ba_{n_1}^*\ba_{n_2}|$ is lower bounded by
  \beq
  \Pr\big(\max_{n_1\neq n_2} |\ba_{n_1}^*\ba_{n_2}| <\delta\big)
  \d4\!&=&\d4\! \Pr(\mu(\bA)\leq \delta) \nn \\
  \d4\!&\geq&\d4\!\! \bigg(\!1-\Big(1-\frac{2}{g}\Big)^{-\frac{1}{2}}e^{-\frac{\delta^2M}{g}}\!\bigg)^{\frac{N(N-1)}{2}}. \nonumber
  \eeq
  This completes the proof.
\end{proof}
\begin{remark}
Because Bernoulli random matrices with each entry filled with $\pm\frac{1}{\sqrt{M}}$ can be regarded as a special case of the RPR matrices in \eqref{random phase-rotated measurements} when $\Theta_{mn}\in\{0,\pi\}$ with equal probability, $\forall m,n$, the coherence statistic in \eqref{Coherence Statistic of RPR} also holds for the Bernoulli random matrix.
\end{remark}


\section{Support Detection Bounds for OMP} 
In this section, the coherence statistics of RPR measurement matrices are applied to obtain the probability bounds of SD for OMP.

\subsection{Measurement Model and OMP Algorithm}
\begin{algorithm}[t]
\caption{OMP for SD} \label{OMP for SSD}
\begin{algorithmic}[1]
\Require
$\bA$, $\by$, and $K$.
\Ensure
$\hat{\cS}$.
\State Initialization: Set iteration number $t=1$, $\br_0=\by$, and $\cS_0=\phi$.
\State Select the active index: $i_{t} = \argmax_{n\in\cS^C_{t-1}} |\ba_n^*\br_{t-1}|$. \label{returning step}
\State Update the active support set: $\cS_t = \cS_{t-1}\cup \{i_{t}\}$. \label{Updated support set}
\State Estimate the signal vector: $\hat{\bx}_t = \argmin\limits_{\bz:\text{supp}(\bz)=\cS_t} \|\by-\bA\bz\|_2^2$. \label{refinement}
\State Update the residual: $\br_t = \by-\bA\hat{\bx}_t = \by-\bA_{\cS_t}\hat{\bx}_{\cS_t}$. \label{updated residual}
\If {$|\cS_t|=K$} terminate and
\Return $\hat{\cS}=\cS_t$.
\Else \ $t=t+1$ and \Return to Step \ref{returning step}.
\EndIf
\end{algorithmic}
\end{algorithm}
Suppose a measurement model
\beq \label{system model}
\by = \bA\bx,
\eeq
where each entry of $\bA\in\C^{M\times N}$ follows \eqref{random phase-rotated measurements}.
Here, the assumption is that the number of measurements $M$ is smaller than the signal dimension $N$, i.e., $M<N$.
The signal $\bx\in\C^{N\times1}$ in \eqref{system model} has $K$ nonzero elements (supports) whose indexes are defined by the support set
\beq \label{support set}
\cS\!=\! \text{supp}(\bx)\!=\! \left\{n_1,\ldots,n_K|x_{n_k}\!\neq\! 0, n_k\!\in\!\{1,\ldots,N\}\!\right\},
\eeq
where $|\cS|=K\ll M$.
The goal is to detect the support set $\cS$ from the measurement $\by\in\C^{M\times 1}$ in \eqref{system model}.

An iterative procedure of OMP for SD is depicted in Algorithm \ref{OMP for SSD} for the measurement model in \eqref{system model}.
To make sure that the active index determined in Step \ref{returning step} is a true support, the following sufficient condition \cite{Tropp04} should be met,
\beq \label{condition of the greedy selection ratio}
\rho(\br_{t-1})\triangleq\frac{\|\bA_{\cS^C}^*\br_{t-1}\|_\infty}{\|\bA_{\cS}^*\br_{t-1}\|_\infty}<1,
\eeq
where $\bA_{\cS}\in\C^{M\times K}$ is the submatrix formed by taking the columns of $\bA$ indexed by $\cS$ and $\bA_{\cS^C}\in\C^{M\times(N-K)}$ is the complementary submatrix of $\bA_{\cS}$.
The nonzero coefficients $\hat{\bx}_{\cS_t}\in\C^{t\times1}$ estimated in Step \ref{updated residual} are formed by extracting the nonzero elements of $\hat{\bx}_t\in\C^{N\times1}$ indexed by $\cS_t$ and given by $\hat{\bx}_{\cS_t}=(\bA_{\cS_t}^*\bA_{\cS_t})^{-1}\bA_{\cS_t}^*\by$. It is crucial to recognize that the updated residual $\br_t$ is orthogonal to the columns of $\bA_{\cS_t}$. The OMP detects one support at each iteration and runs for exactly $K$ iterations.

\subsection{Support Detection Performance Guarantee}
We provide the SD performance guarantee of the OMP in Algorithm \ref{OMP for SSD} as follows.
\begin{theorem} \label{Main Result 2}
  Suppose the measurement model in \eqref{system model} with the RPR measurement matrix $\bA$ based on \eqref{random phase-rotated measurements}. Then, the OMP in Algorithm \ref{OMP for SSD} detects the $K$ supports of $\bx$ for any $(M,N)$ with
 \beq
\Pr(\cV_{\text{SSD}})\geq 1-\Big(1-\frac{2}{g}\Big)^{-\frac{1}{2}}KN\cdot e^{-\frac{M}{gK^2}},~g>2, \label{final lower bound of Pssd:RPR}
\eeq
where $\cV_{\text{SSD}}$ is the event of successful SD (SSD) after $K$ iterations.
When the number of measurements $M$ satisfies
\beq \label{scaling law of RPR}
 M\geq gK^2 \ln\bigg( \frac{KN}{\epsilon\sqrt{1-\frac{2}{g}}}\bigg),\ g>2,
\eeq
for $\epsilon\in(0,1)$, Algorithm \ref{OMP for SSD} satisfies $\Pr(\cV_{\text{SSD}}) \geq  1-\epsilon$.
\end{theorem}

\quad\textit{Proof:} See Appendix \ref{Proof of Theorem 2}.

To further tighten the lower bound in \eqref{scaling law of RPR}, we optimize the free variable $g$ by minimizing the right hand side (r.h.s.) of \eqref{scaling law of RPR} such that
\beq
g^{\text{opt}} = \argmin_{g>2} f(g) \triangleq  gK^2 \ln\bigg( \frac{KN}{\epsilon\sqrt{1-\frac{2}{g}}}\bigg). \label{g_opt}
\eeq
\begin{theorem} \label{theorem: g^{opt}}
	The objective function $f(g)$ in \eqref{g_opt} is convex for $g>2$ and a closed-form expression of $g^{\text{opt}}$ is given by
	\beq
	g^{\text{opt}} = \frac{2}{1+ \big( W_{-1}( - ( \frac{\epsilon}{K N} )^2  e^{-1})  \big)^{-1} },  \label{eq: g^{opt}}
	\eeq
	where $W_{-1}(\cdot)$ is the lower branch of the Lambert $W$ function \cite{gradshteyn2007}, defined by $z = W_{-1}(z e^z)$ for $z < -1$.
\end{theorem}

\quad\textit{Proof:} See Appendix \ref{Proof of Theorem 3}.


\section{Numerical Simulations}
\begin{figure}[t]
	\centering
	\includegraphics[width=9cm, height=7.5cm]{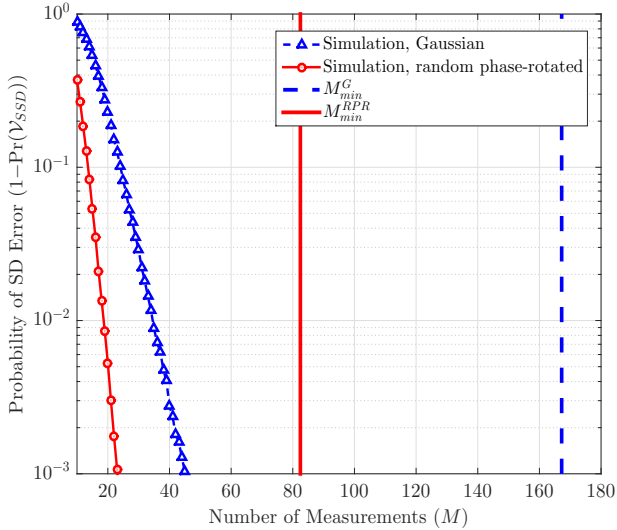}
	\caption{SD performance guarantees of OMP with the RPR and Gaussian measurements when $N=200$, $K=2$, $\epsilon=10^{-1}$, and $g^{\text{opt}}=2.1020$. } \label{Fig:Tightness evaluation}
\end{figure}
To verify the SD performance guarantee in \eqref{scaling law of RPR}, we perform Monte Carlo simulations in Fig. \ref{Fig:Tightness evaluation}, where the probability of SD error, i.e., $1-\Pr(\cV_{\text{SSD}})$, across different numbers of measurements $M$ for $N=200$ and $K=2$, is evaluated. In the simulation, the signal $\bx$ is generated by randomly choosing $K$ supports with each support having $x_n=1$, for $n\in\cS$, and we compare with the existing coherence-based SD performance guarantee for the Gaussian random measurement matrix \cite{Tropp07}. In Fig. \ref{Fig:Tightness evaluation}, the vertical lines denote the minimum required $M$ to guarantee the SD error rate $\epsilon=10^{-1}$, where these values are given by the r.h.s. of \eqref{scaling law of RPR} for the RPR measurements $(M_{min}^{RPR}=82~\text{with}~ g^{\text{opt}}=2.1020~\text{according to}~\eqref{eq: g^{opt}})$, and $M\geq CK\ln(N/\epsilon)$ for the Gaussian case $(M_{min}^{G}=168~\text{with}~C=11)$\cite{Tropp07}, respectively.
Seen from Fig. \ref{Fig:Tightness evaluation}, the obtained SD performance guarantee of RPR matrices provides a tighter characterization than the Gaussian case when the signal is sparse, i.e., $K$ is small.


\section{Conclusion and Discussion}
The coherence statistics of RPR matrices were analyzed and applied to obtain the SD performance guarantees of OMP. The introduced free variable was optimized to further tighten the SD bound. Numerical simulations corroborated the theoretical findings and revealed that including the constant modulus and unit-norm structure for random measurement ensembles is desirable for SD using OMP.

In this work, we focused on the coherence statistics of RPR matrices to show the SD performance guarantees of OMP. In particular, we proved that OMP can achieve SSD with high probability, provided $M=O(K^2\ln(KN))$ RPR measurements. It is of interest to compare our coherence-based analysis with the restricted isometry property (RIP)-based result since they are two main techniques in analyzing the performance guarantees of SD for OMP.
By using the concentration inequality in \cite[Theorem 2]{Zhang18} and the method of proving the RIP for random matrices in \cite[Theorem 5.2]{Baraniuk08}, one can obtain that $M\geq 16K\ln(N/K)/\delta^2$ is sufficient for the RPR matrices to satisfy the RIP with high probability, where $\delta\in(0,1)$ is the restricted isometry constant.
With $\delta<\frac{1}{\sqrt{K}}$ being a strict condition of SSD for OMP \cite{Davenport10}, the RIP-based SD bound can be given by $M\geq 16K^2\ln(N/K)$, which is on par with our coherence-based results in Theorem \ref{Main Result 2}.

Finally, one limitation of the work is that the SD bound becomes loose as $K$ grows. Seen from Fig. \ref{Fig:Tightness evaluation}, there is still room for further improvement by investigating a new structure of random measurement ensembles, which is subject to future research.


\appendices
\section{Proof of Lemma \ref{RPR and Bernoulli vectors} } \label{Proof of Lemma 1}
\begin{proof}
The left hand side (l.h.s.) and r.h.s. of \eqref{inequality between RPR and Bernoulli vectors} can be rewritten as $\E[|\bp^*\bu|^{2k}]=\E[|\sum_{m=1}^{M}u_m e^{-j\theta_m}|^{2k}]/{M^{k}}$ and $\mathbb{E}[|\bq^*\bar{\bu}|^{2k}]=\E[|\sum_{m=1}^{M}\zeta_m|^{2k}]/M^{2k}$, respectively,
where $\zeta_m\in\{1,-1\}$, $\forall m$, with equal probability.
Thus, showing the inequality in \eqref{inequality between RPR and Bernoulli vectors} is equivalent to showing
\beq \label{inequality-modified}
M^k\E\bigg[\Big|\sum_{m=1}^{M}u_m e^{-j\theta_m}\Big|^{2k}\bigg]\leq
\E\bigg[\Big|\sum_{m=1}^{M}\zeta_m\Big|^{2k}\bigg].
\eeq

The l.h.s. of \eqref{inequality-modified} can be simplified as
\beq 
\d4\d4\d4\d4\d4&&M^k\E\bigg[\Big|\sum_{m=1}^{M}u_m e^{-j\theta_m}\Big|^{2k}\bigg] \nn \\
\d4\d4\d4\d4\d4&&\overset{(a)}{=}M^k\E\bigg[\Big(\big|\sum_{m=1}^{M}u_m\cos(\theta_m)\big|^2+\big|\sum_{m=1}^{M}u_m\sin(\theta_m)\big|^2\Big)^k\bigg] \nn \\
\d4\d4\d4\d4\d4&&\overset{(b)}{=}M^k\E\bigg[\Big(\sum_{m=1}^{M}|u_m|^2\Big)^k\bigg]\overset{(c)}{=}M^k, \label{LHS}
\eeq
where $(a)$ follows from the equality  $e^{-j\theta_m}=\cos(\theta_m)-j\sin(\theta_m)$, $(b)$ is due to the fact that $\E[\cos(\theta_{m_1})\cos(\theta_{m_2})]=\E[\sin(\theta_{m_1})\sin(\theta_{m_2})]=0$ for $m_1\neq m_2$, and $(c)$ holds because $\|\bu\|_2=1$.
Expanding the r.h.s. of \eqref{inequality-modified} leads to
\beq
\d4\d4\E\Big[\big|\sum_{m=1}^{M}\zeta_m\big|^{2k}\Big]
\d4&=&\d4\E\big[(M+G(M))^k\big] \nn \\
\d4&=&\d4\E\bigg[\sum_{i=0}^{k}\binom{k}{i}M^{k-i}G(M)^i\bigg]\geq M^k, \label{RHS}
\eeq
where $G(M)\triangleq\sum_{m_1=1}^{M}\sum_{m_2=1,m_2\neq m_1}^{M}\zeta_{m_1}\zeta_{m_2}$. The inequality in \eqref{RHS} becomes the equality only if $k=0,1$ because $\E[\zeta_{m_1}\zeta_{m_2}]=0$ for $m_1\neq m_2$. On the other hand, when $k>1$, the strict inequality in \eqref{RHS} holds because $\E[\zeta_{m_1}^{2l_1}\zeta_{m_2}^{2l_2}]=1$ for any positive integers $l_1$, $l_2$, leading to $\E[G(M)^i]>0$ for $i>1$.
Combining \eqref{LHS} and \eqref{RHS} results in \eqref{inequality-modified}. 
\end{proof}

\section{Proof of Lemma \ref{Probability Upperbound for inner product}} \label{Proof of Lemma 2}
\begin{proof}
By using Markov's inequality, we have for $h\geq 0$,
\beq \label{Markov's inequality}
\Pr(|\bp^*\bu|\geq\delta)\!=\!\Pr(|\bp^*\bu|^2\geq \delta^2)\!\leq\! \mathbb{E}\left[e^{h|\bp^*\bu|^2}\right]e^{-h\delta^2}.
\eeq
The term $\mathbb{E}[e^{h|\bp^*\bu|^2}]$ in \eqref{Markov's inequality} can further be upper bounded for $h\in[0,~M/2)$ by
\beq \label{UpperBound of Exponential}
\mathbb{E}\big[e^{h|\bp^*\bu|^2}\big]\leq
\mathbb{E}\big[e^{h|\bq^*\bar{\bu}|^2}\big]\leq \Big(1-\frac{2h}{M}\Big)^{-\frac{1}{2}},
\eeq
where the first inequality is due to the Taylor series expansion of
$\mathbb{E}[e^{h|\bp^*\bu|^2}]=$
$\sum_{k=0}^{\infty}\frac{h^k}{k!}\mathbb{E}[|\bp^*\bu|^{2k}]$
and Lemma \ref{RPR and Bernoulli vectors} applied to $\mathbb{E}[|\bp^*\bu|^{2k}]$, and $\bar{\bu}$ follows the same definition in Lemma \ref{RPR and Bernoulli vectors}. The last step in \eqref{UpperBound of Exponential} follows from the inequality $\E[e^{h|\bq^*\bu|^2}]\leq 1/\sqrt{1-2h/M}$ for $h\in[0,~M/2)$ in \cite[Lemma 5.2]{Achlioptas03}.

Inserting \eqref{UpperBound of Exponential} into \eqref{Markov's inequality} leads to
\beq \label{intermediate bound}
\Pr(|\bp^*\bu|\geq \delta) \leq \Big(1-\frac{2h}{M}\Big)^{-\frac{1}{2}}e^{-h\delta^2}.
\eeq
Because the inequality holds for any $h\in[0,M/2)$, substituting $h=M/g$, $g>2$, into \eqref{intermediate bound} completes the proof.
\end{proof}

\section{Proof of Theorem \ref{Main Result 2}} \label{Proof of Theorem 2}
\begin{proof}
The proof is inspired by a similar theorem in \cite[Theorem 6]{Tropp07} and refines the results for the RPR measurement ensembles in conjunction with Lemma \ref{Probability Upperbound for inner product} and Theorem \ref{lower bound of coherence statistic of RPR}.
We first elaborate two events:
1) $\cV_{\text{SSD}}$ is defined on the basis of the condition in \eqref{condition of the greedy selection ratio} as $\cV_{\text{SSD}}\triangleq \{\max_{t=1,\ldots,K} \rho(\br_{t-1})=\frac{\|\bA_{\cS^C}^*\br_{t-1}\|_\infty}{\|\bA_{\cS}^*\br_{t-1}\|_\infty}<1\}$;
and 2) The event that $\mu(\bA_{\cS})$ is bounded by $1/K$, i.e.,
$\mathcal{D} \triangleq \{\mu(\bA_{\cS})< 1/K\}$.
The event $\cD$ is to restrict the $\cV_{\text{SSD}}$ on a special class of $\bA$ to ease the bound analysis below.

Conditioned on the event $\cD$, the probability of SSD can be lower bounded by
\beq \label{Pssd}
\Pr(\cV_{\text{SSD}})\geq \Pr(\cV_{\text{SSD}}\cap\cD)=\Pr(\cV_{\text{SSD}}|\cD)\Pr(\cD).
\eeq
From Theorem \ref{lower bound of coherence statistic of RPR}, $\Pr(\cD)$ in \eqref{Pssd} can be lower bounded by
\beq \label{lower bound of mutual coherence of RPR}
  \Pr(\cD) \d4&=&\d4 \Pr\Big(\mu(\bA_{\cS})\leq \frac{1}{K}\Big) \nn \\
  \d4&\geq &\d4 \bigg[1-\Big(1-\frac{2}{g}\Big)^{-\frac{1}{2}}e^{-\frac{M}{gK^2}}\bigg]^{\frac{K(K-1)}{2}},
\eeq
where $g>2$.
The conditional probability on the r.h.s. of \eqref{Pssd} can be lower bounded by
\beq
\d4\d4\Pr(\cV_{\text{SSD}}|\cD) \d4&=&\d4 \Pr\bigg(\max_t \frac{\|\bA_{\cS^C}^*\br_{t-1}\|_\infty}{\|\bA_{\cS}^*\br_{t-1}\|_\infty}<1 \Big|\cD\bigg) \nn \\
\d4&\overset{(a)}{\geq}&\d4 \Pr\bigg(\max_t \frac{\sqrt{K}\max_{j\in\cS^C} |\ba_j^*\br_{t-1}|}{\|\bA_{\cS}^*\br_{t-1}\|_2}<1 \Big|\cD\bigg) \nn  \\
\d4&\overset{(b)}{\geq}&\d4 \Pr\Big(\max_t\ \max_{j\in\cS^C} |\ba_j^*\bb_{t-1}|< \frac{1}{K} \big|\cD\Big) \nn \\
\d4&\overset{(c)}{=}&\d4 \prod_{j\in\cS^C}  \Pr\Big(\max_t\ |\ba_j^*\bb_{t-1}|< \frac{1}{K} \big|\cD\Big) \nn \\
\d4&\overset{(d)}{\geq}&\d4 \bigg[1-\Big(1-\frac{2}{g}\Big)^{-\frac{1}{2}}e^{-\frac{M}{gK^2}}\bigg]^{K(N-K)} \label{Pconditional of RPR}
\eeq
where $(a)$ is due to the inequality $\|\bu\|_\infty\geq \|\bu\|_2/\sqrt{K}$ for $\bu\in\C^{K\times1}$, $(b)$ comes from $\bb_{t-1}\triangleq \tilde{\bb}_{t-1}/\|\tilde{\bb}_{t-1}\|_2$ where $\tilde{\bb}_{t-1}\triangleq \br_{t-1}/(\sqrt{K}\|\bA_{\cS}^*\br_{t-1}\|_2)$ and $\|\tilde{\bb}_{t-1}\|_2\leq 1$ because $\|\bA_{\cS}^*\br_{t-1}\|_2/\|\br_{t-1}\|_2\geq \sqrt{\lambda_{min}(\bA_{\cS}^*\bA_{\cS})}\geq \sqrt{1-(K-1)\mu(\bA_{\cS})}\geq 1/\sqrt{K}$ by applying Gershgorin disc theorem \cite{Golub1996}, $(c)$ holds due to the fact that the $N-K$ columns of $\bA_{\cS^C}$ are independent, and $(d)$ is due to Lemma \ref{Probability Upperbound for inner product}.

Substitute \eqref{lower bound of mutual coherence of RPR} and \eqref{Pconditional of RPR} into \eqref{Pssd} yields
  \beq
  \Pr(\cV_{\text{SSD}})\d4&\geq&\d4 \bigg[1-\Big(1-\frac{2}{g}\Big)^{-\frac{1}{2}}e^{-\frac{M}{gK^2}}\bigg]^{K(N-K)+\frac{K(K-1)}{2}} \nonumber \\
                       \d4&\overset{(a)}{\geq}&\d4 \! 1\!-\!\Big(1\!-\!\frac{2}{g}\Big)^{-\frac{1}{2}}\!\Big[K(N-K)\!+\!\frac{K(K-1)}{2}\Big]e^{-\frac{M}{gK^2}} \nonumber \label{Relaxation} \\
                       \d4&\geq&\d4 1-\Big(1-\frac{2}{g}\Big)^{-\frac{1}{2}}KNe^{-\frac{M}{gK^2}}, \nn
  \eeq
  where $(a)$ holds because $(1-2/g)^{-1/2}e^{-\frac{M}{gK^2}}<1$ and $K(N-K)+K(K-1)/2>1$.
  Setting $(1-2/g)^{-1/2}KNe^{-\frac{M}{gK^2}}\leq\epsilon$ and taking the natural logarithm of both sides reveals that $\Pr(\cV_{\text{SSD}})\geq 1-\epsilon$ when $M\geq gK^2\ln(KN/(\epsilon\sqrt{1-2/g}))$. 
\end{proof}

\section{Proof of Theorem \ref{theorem: g^{opt}} } \label{Proof of Theorem 3}
\begin{proof}
We first claim that the objective function $f(g)$ in \eqref{g_opt} is convex for $g>2$.
To show this, we check the second-order condition $f^{\prime\prime}(g) >0$, where $f^{\prime\prime}(g)$ is the second-order derivative of $f(g)$.
After some algebraic manipulations, the first and second-order derivatives of $f(g)$ can be written, respectively, as
$f^\prime(g) = K^2\ln(\frac{KN}{\epsilon\sqrt{1-\frac{2}{g}}})-\frac{K^2}{g-2}$
and
$f^{\prime\prime}(g) = \frac{2K^2}{g(g-2)^2}$.
Because $f^{\prime\prime}(g) >0$  for $g>2$, $f(g)$ is convex.

The optimality condition of \eqref{g_opt} can now be described by using the first-order optimality condition $f^\prime(g^{\text{opt}}) = 0$ as
\beq \label{eq: optimality condition}	
f(g^{\text{opt}})= \frac{g^{\text{opt}}}{g^{\text{opt}}-2}K^2.
\eeq
Let $\al = 1-\frac{2}{g^{\text{opt}}}$, equivalently $\frac{1}{g^{\text{opt}}-2} = \frac{1-\al}{2\al}$.
Then, by \eqref{g_opt}, the equality in \eqref{eq: optimality condition} can be rewritten as $\lp \frac{\epsilon}{K N} \rp^2 e^{-1}  =   \frac{1}{\al}  e^{-\frac{1}{\al}}$.
This yields $\al = - \lp W_{-1}\lp - \lp \frac{\epsilon}{K N} \rp^2  e^{-1}  \rp \rp^{-1}$, which follows from the definition of the lower branch of the Lambert $W$ function $W_{-1}( -\frac{1}{\al} e^{-\frac{1}{\al}} ) = -\frac{1}{\al}$ and $\al < 1$ \cite{gradshteyn2007}.
Now, by the equality $g^{\text{opt}} = \frac{2}{1-\al}$, we finally have \eqref{eq: g^{opt}}. This completes the proof.
\end{proof}


\bibliographystyle{IEEEtran}
\bibliography{IEEEabrv,reference}

\end{document}